\documentclass[conference]{IEEEtran}
\usepackage{amsmath,amssymb,amsfonts}
\usepackage{algorithmic}
\usepackage{cite}
\usepackage{graphicx}
\usepackage{textcomp}
\usepackage{xcolor}
\def\BibTeX{{\rm B\kern-.05em{\sc i\kern-.025em b}\kern-.08em
    T\kern-.1667em\lower.7ex\hbox{E}\kern-.125emX}}

\usepackage{algorithm}

\usepackage{mathtools}

\usepackage{amsthm}
\newtheorem{theorem}{Theorem}
\newtheorem{example}{Example}
\newtheorem{proposition}{Proposition}
\newtheorem{lemma}{Lemma}

\usepackage{multirow}

\DeclareMathOperator*{\argmax}{argmax}

\begin{document}

%
%

\title{Transitions from P to NP-hardness: the case of the Linear Ordering Problem\\


\thanks{Identify applicable funding agency here. If none, delete this.}
}

\author{\IEEEauthorblockN{Anne Elorza}
\IEEEauthorblockA{\textit{Department of Computer Science and Artificial Intelligence} \\
\textit{University of the Basque Country UPV/EHU}\\
San Sebastian, Spain \\
anne.elorza@ehu.eus}
\and
\IEEEauthorblockN{Leticia Hernando}
\IEEEauthorblockA{\textit{Department of Mathematics} \\
\textit{University of the Basque Country UPV/EHU}\\
Leioa, Spain \\
leticia.hernando@ehu.eus}
\and
\IEEEauthorblockN{Jose A. Lozano}
\IEEEauthorblockA{\textit{Basque Center for Applied Mathematics (BCAM)} \\
Bilbao, Spain}
\IEEEauthorblockA{\textit{Department of Computer Science and Artificial Intelligence} \\
\textit{University of the Basque Country UPV/EHU}\\
San Sebastian, Spain \\
ja.lozano@ehu.eus}
}

\maketitle

%
%

\begin{abstract}
In this paper we evaluate how constructive heuristics degrade when a problem transits from P to NP-hard. This is done by means of the linear ordering problem. More specifically, for this problem we prove that the objective function can be expressed as the sum of two objective functions, one of which is associated with a P problem (an exact polynomial time algorithm is proposed to solve it), while the other is associated with an NP-hard problem. We study how different constructive algorithms whose behaviour only depends on univariate information perform depending on the contribution of the P or NP-hard components of the problem. A number of experiments are conducted with reduced dimensions, where the global optimum of the problems is known, giving different weights to the NP-hard component, while the weight of the P component is fixed. It is observed how the performance of the constructive algorithms gets worse as the weight given to the NP-hard component increases.
\end{abstract}


%
%

\section{Introduction}

The area of combinatorial optimization is composed of those optimization problems whose search space is finite or countably infinite. Among these problems, the most prominent ones are those that have been classified as NP-hard. For these problems, there is no known algorithm that can solve all of their instances in polynomial time on the size of the problem. However, an NP-hard problem, such as the Quadratic Assignment Problem \cite{commander2005survey}, can have particular cases that can be solved in polynomial time \cite{erdougan2011two}. It could also happen that, depending on the value of a given parameter that is used to define a problem instance, the problem could become more or less easy to be solved. This has been extensively studied in the area of parametric complexity. In this context, the concept of phase transition has been incorporated from the physical field to the field of combinatorial optimization. Originally, a phase transition is understood as a dramatic change of a property of a system when a given parameter crosses a specific threshold, such as what happens when the temperature of water falls below 0º and freezes. In the area of combinatorial optimization, computational complexity transitions have been studied, paying special attention to two famous problems: the traveling salesman problem and the satisfiability problem \cite{gent1996tsp,hernando2011study,hartmann2006phase,hogg1996phase,zhang1996study,biroli2002phase}. The study of phase transitions generally results in a deeper understanding of the problems.   

Our study is focused on an optimization problem based on permutations, the linear ordering problem (LOP), which has also been classified as NP-hard \cite{marti2011linear}. In this paper, we decompose the LOP into the sum of two components. This is carried out by  computing its Fourier transform (FT) over the symmetric group. Similarly to its real counterpart, which is widely known and decomposes a real function into an infinite sum of sines and cosines, the FT over permutations decomposes a permutation-based function in terms of a family of basis functions, which are named the irreducible representations of the symmetric group. The permutation-based FT has been applied in the area of combinatorial optimization where the Fourier coefficients of the fitness functions of three famous problems (the LOP, the Traveling Salesman Problem and the Quadratic Assignment Problem) have been characterized \cite{elorza2019taxonomization,kondor2010fourier}. We show how the Fourier characterization of the LOP naturally leads to a decomposition of its fitness function as a sum of two functions. One of them contains the whole univariate information of the problem, that is, the information about an element located at a certain position, independently of which exact position the rest of the elements occupy (this concept is explained in Section \ref{sec::decomposition_LOP}). The other function contains the rest of the information of the problem. It turns out that the first function is associated with a P problem, while the other is associated with an NP-hard problem. 

Taking this theoretical decomposition as a basis, we observe the evolution of the performance of a number of constructive algorithms when they are applied to instances as different weights are given to the P and NP-hard components. One of the constructive algorithms used is the well-known Becker's method \cite{becker1967helmstadtersche,laguna1999intensification}. In addition, we also design three new constructive techniques. It is worth noting  that  all  of  the four  constructive algorithms  that  we  are  using  employ univariate information into their machinery, which is the same type of information as the one that the P component of our decomposition contains.

The paper is organized as follows: Section \ref{sec::FT_LOP} outlines the mathematical background, defining the LOP and describing the Fourier characterization of the problem. Section \ref{sec::decomposition_LOP} shows the decomposition of the LOP into its P and NP-hard components. The constructive algorithms considered are presented in Section \ref{sec::constructives}, and the experimental study is undertaken in Section \ref{sec::experimental_study}. Finally, the conclusions and future work are drawn in Section \ref{sec::conclusions}.

%
%

\section{Fourier Transform of the LOP}\label{sec::FT_LOP}

%
%

The LOP is a problem that belongs to the field of permutation-based combinatorial optimization. Given a square matrix $A = [a_{ij}]$ of size $n$, the LOP consists of finding the joint permutations of rows and columns that maximize the sum of the upper-diagonal elements \cite{marti2011linear}. Then, the aim is to maximize the following function:
\begin{equation}\label{eq::LOP_function_definition}
f(\sigma) = \sum_{i=1}^{n-1} \sum_{j=i+1}^n\!\! a_{\sigma(i)\sigma(j)} \hspace{0.5cm} \text{ with } \sigma \in \Sigma_n,
\end{equation}
where $\Sigma_n$ is the set of permutations of size $n$ and $\sigma(k)$ denotes the number of the row/column of the original matrix which is located in the $k$-th position. We assume throughout the rest of the paper that the diagonal elements of $A$ are zero-valued.

The LOP can alternatively be seen in Fourier space through the Fourier coefficients of its objective functions. The FT over the symmetric group and the Fourier characterization of the LOP are briefly introduced in this section (more information can be found in \cite{huang2009fourier,elorza2019taxonomization,sagan2013symmetric,kondor2010fourier}).

\subsection{FT over the symmetric group}

 The FT over the symmetric group is a mathematical tool that decomposes a function $f:\Sigma_n\longrightarrow\mathbb{R}$ with respect to an orthogonal set of basis functions. In this case, the family of basis functions is the set of irreducible representations of the symmetric group, $\{\rho_\lambda: \lambda\vdash n\}$, where $\lambda$ denotes a partition of $n$ (a tuple that adds up to $n$). The Fourier coefficients are then indexed by the partitions of $n$ and can be computed according to the following formula:
$$\hat{f}_{\rho_\lambda} = \sum_{\sigma\in\Sigma_n} f(\sigma)\rho_\lambda(\sigma).$$

In what follows, the notation is simplified, denoting the Fourier coefficients as $\hat{f}_\lambda$. It is also known that there is a bijection between Fourier coefficients and permutation-based functions, which means that a function can be determined by the collection of its Fourier coefficients.

These coefficients have a highly intuitive interpretation when the function is a probability $p$. In this context, the lowest Fourier coefficient  satisfies $\hat{p}_{(n)} = 1$. On the other hand, $\hat{p}_{(n-1,1)}$ stores information of first order marginals (statements of the form $p(\sigma:\sigma(i) = j)$), while higher Fourier coefficients store information of higher order marginals when the effect of lower order marginals is subtracted. For example, $\hat{p}_{(n-2,2)}$ stores information about statements of the form $p(\sigma:\sigma(\{i,j\}) = \{k,l\})$ when the effect of marginals of order one is subtracted. This means that, if one knows both $\hat{p}_{(n-1,1)}$ and $\hat{p}_{(n-2,2)}$, then one can compute $p(\sigma:\sigma(\{i,j\}) = \{k,l\})$ for any $i,j,k,l$. A detailed explanation of this interpretation can be found in \cite{huang2009fourier}.

\subsection{FT of the LOP}

Given a permutation-based combinatorial optimization problem, one could analyse the shape of the Fourier coefficients of the objective functions associated with the problem. In the case of the LOP, this study has been carried out in \cite{elorza2019taxonomization}, where the following theorem has been proved:

\begin{theorem}\label{theo::LOPcoeff}

If $f:\Sigma_n\longrightarrow \mathbb{R}$ is the objective function of an LOP instance, that is, $f$ is expressed as in (\ref{eq::LOP_function_definition}), then its FT has the following properties:
\begin{enumerate}
\item $\hat{f}_\lambda = 0$, if $\lambda\ne (n), (n-1,1), (n-2,1,1)$.
\item $\hat{f}_\lambda$ has at most rank one for $\lambda = (n-1,1), (n-2,1,1)$. Having rank one is equivalent to the fact that the matrix columns are proportional. 
\item For $\lambda = (n-1,1), (n-2,1,1)$ and a fixed dimension $n$, the proportions among the columns of $\hat{f}_\lambda$ are the same for all the instances.
\end{enumerate}
\end{theorem}

This means that an objective function $f$ of the LOP has, at most, three non-zero Fourier coefficients: $\hat{f}_{(n)}$, $\hat{f}_{(n-1,1)}$ and $\hat{f}_{(n-2,1,1)}$ . The effect of the first coefficient in terms of optimization is irrelevant, because it is related to the mean value of $f$, but, when one dives into a deeper analysis of the other two components, a couple of interesting properties come to light.

%
%

\section{Decomposition of the LOP}\label{sec::decomposition_LOP}

This section is devoted to proving that an objective function of the LOP can be decomposed into the sum of two objective functions, one of which is associated with a P problem, while the other is associated with an NP-hard problem. For ease of reading, most of the mathematical proofs are developed in Appendix \ref{app:proofs_decomposition_LOP}, where any interested reader may refer to gain a deeper insight into the matter. In the rest of the paper, we may refer to an objective function of an LOP instance simply as LOP function.


\subsection{P component of the LOP}

Let us focus on the case where an LOP function $f$ satisfies $\hat{f}_{(n-2,1,1)} = 0$. Intuitively, if $f$ is a probability $p$, and only coefficients $(n-1,1)$ and $(n)$ are non-zero, then, $f$ can be reconstructed through marginals of order 1 (as explained in \cite{huang2009fourier}). Note that this expression belongs to the area of probability. If a marginal of order 1 in this area is a statement of the form 
$$p(\sigma:\sigma(i) = j) = \sum_{\sigma:\sigma(i) = j} p(\sigma),$$
then, for a general function $f$, the univariate information can be translated as: $$\sum_{\sigma:\sigma(i)=j} f(\sigma).$$ Thus, a function $f$ with only the two lowest Fourier coefficients, $\hat{f}_{(n-1,1)}$ and $\hat{f}_{(n)}$, different from zero, can be reconstructed through this univariate information. Given a permutation $\sigma = [i_1\ i_2\ \cdots i_n]\in\Sigma_n$, this idea can be expressed according to the following formula:

\begin{align}\label{eq::function_only_2lowestCoefficients}
& f([i_1\ i_2\ \cdots\ i_n]) = \frac{1}{n!} \Bigg( \sum\limits_{\sigma:\sigma(1) = i_1} f(\sigma) + \sum\limits_{\sigma:\sigma(2) = i_2} f(\sigma) + \nonumber \\
&\cdots + \sum\limits_{\sigma:\sigma(n) = i_n} f(\sigma) \Bigg).
\end{align}

Based on this equation, we can characterize the input matrix of an LOP function $f$ such that $\hat{f}_{(n-2,1,1)} = 0$.

\begin{proposition}\label{prop::LOP_necessaryConditionA_n-2_1_1_zero}

The objective function $f:\Sigma_n\longrightarrow\mathbb{R}$ of an LOP instance with an input matrix $A$ satisfies $\hat{f}_{(n-2,1,1)} = 0$ if and only if 
$$ a_{ij} - a_{ji} + a_{jk} - a_{kj} = a_{ik} - a_{ki} \qquad\forall i,j,k=1,2,\ldots,n.$$
\end{proposition}

Let us consider the LOP functions $f$ such that $\hat{f}_{(n-2,1,1)} = 0$. If one considers swaps among the elements of permutations, this type of functions have a property that cannot be generalized to any LOP function. We illustrate this with an example.

\begin{example}
With a fixed dimension of $n=4$, we compute the difference between objective function values after applying a swap of the elements 1 and 2. We start with 3 examples. The first one shows what happens when we perform an adjacent swap:
$$f([\mathbf{1}\ \mathbf{2}\ 3\ 4]) - f([\mathbf{2}\ \mathbf{1}\ 3\ 4]) = a_{12}-a_{21}.$$

Secondly, we perform a swap where there is only one element between the elements 1 and 2.
$$f([\mathbf{1}\ 3\ \mathbf{2}\ 4]) - f([\mathbf{2}\ 3\ \mathbf{1}\ 4]) = a_{12}-a_{21} + a_{13}-a_{31} + a_{32}-a_{23}. $$

According to Proposition \ref{prop::LOP_necessaryConditionA_n-2_1_1_zero}, $a_{13}-a_{31} + a_{32}-a_{23}=a_{12}-a_{21}$. Then,
$$f([\mathbf{1}\ 3\ \mathbf{2}\ 4]) - f([\mathbf{2}\ 3\ \mathbf{1}\ 4]) = a_{12}-a_{21} + a_{12}-a_{21} = 2(a_{12}-a_{21}). $$

Thirdly, we perform a swap where there are 2 elements between elements 1 and 2:
\begin{align*}
& f([\mathbf{1}\ 3\ 4\ \mathbf{2}]) - f([\mathbf{2}\ 3\ 4\ \mathbf{1}]) = a_{13}-a_{31} + a_{14}-a_{41} + \\& a_{12}-a_{21} + a_{32}-a_{23} + a_{42}-a_{24}.
\end{align*}

We take the following equations into account:
\begin{align*}
& a_{13}-a_{31}  + a_{32}-a_{23} = a_{12}-a_{21} \\
& a_{14}-a_{41} + a_{42}-a_{24} = a_{12}-a_{21} .
\end{align*}

Then,
\begin{align*}
& f([\mathbf{1}\ 3\ 4\ \mathbf{2}]) - f([\mathbf{2}\ 3\ 4\ \mathbf{1}]) = a_{12}-a_{21} + a_{12}-a_{21} + \\& a_{12}-a_{21} = 3(a_{12}-a_{21}). 
\end{align*}

In general, the following can be proved:
\begin{align*}
 f([\mathbf{1}\ \mathbf{2}\ 3\ 4]) - f([\mathbf{2}\ \mathbf{1}\ 3\ 4]) &= a_{12}-a_{21}.\\
 f([\mathbf{1}\ \mathbf{2}\ 4\ 3]) - f([\mathbf{2}\ \mathbf{1}\ 4\ 3]) &= a_{12}-a_{21}.\\
 f([3\ \mathbf{1}\ \mathbf{2}\ 4]) - f([3\ \mathbf{2}\ \mathbf{1}\ 4]) &= a_{12}-a_{21}.\\
 f([4\ \mathbf{1}\ \mathbf{2}\ 3]) - f([4\ \mathbf{2}\ \mathbf{1}\ 3]) &= a_{12}-a_{21}.\\
 f([3\ 4\ \mathbf{1}\ \mathbf{2}]) - f([3\ 4\ \mathbf{2}\ \mathbf{1}]) &= a_{12}-a_{21}.\\
 f([4\ 3\ \mathbf{1}\ \mathbf{2}]) - f([4\ 3\ \mathbf{2}\ \mathbf{1}]) &= a_{12}-a_{21}.\\[4mm]
 f([\mathbf{1}\ 3\ \mathbf{2}\ 4]) - f([\mathbf{2}\ 3\ \mathbf{1}\ 4]) &= 2(a_{12}-a_{21}).\\
 f([\mathbf{1}\ 4\ \mathbf{2}\ 3]) - f([\mathbf{2}\ 4\ \mathbf{1}\ 3]) &= 2(a_{12}-a_{21}).\\
 f([3\ \mathbf{1}\ 4\ \mathbf{2}]) - f([3\ \mathbf{2}\ 4\ \mathbf{1}]) &= 2(a_{12}-a_{21}).\\
 f([4\ \mathbf{1}\ 3\ \mathbf{2}]) - f([4\ \mathbf{2}\ 3\ \mathbf{1}]) &= 2(a_{12}-a_{21}).\\[4mm]
 f([\mathbf{1}\ 3\ 4\ \mathbf{2}]) - f([\mathbf{2}\ 3\ 4\ \mathbf{1}]) &= 3(a_{12}-a_{21}).\\
 f([\mathbf{1}\ 4\ 3\ \mathbf{2}]) - f([\mathbf{2}\ 4\ 3\ \mathbf{1}]) &= 3(a_{12}-a_{21}).
\end{align*}

\end{example}

We generalize the results of this example to any dimension $n$ and the swap between two arbitrary elements $i$ and $j$.

\begin{proposition}\label{prop:LOP_n-2_1_1_zero_swap}
Given the objective function of an LOP instance $f$ such that $\hat{f}_{(n-2,1,1)} = 0$ and a transposition $\sigma_{ij} = (i\ j)$ with $i<j$, then
$$ f(\sigma) - f(\sigma_{ij}\circ\sigma) = (j-i)(a_{\sigma(i)\sigma(j)} - a_{\sigma(j)\sigma(i)}).$$
\end{proposition}

This property has an interesting consequence. If $\sigma_{max}$ is a global maximum and $a_{ij} - a_{ji} > 0$ for two indices $i$ and $j$, then, element $i$ must be before $j$ in the configuration of $\sigma_{max}$. Otherwise, if we applied a swap between elements $i$ and $j$, we would obtain a solution better than $\sigma_{max}$. This implies that $\sigma_{max}$ (the global maximum) can be constructed just by checking whether $a_{ij} - a_{ji} > 0$. In doing so, we can know, for each pair of elements $i$ and $j$, which goes before the other in the configuration of $\sigma_{max}$. These comparisons can be done in polynomial time. Therefore, finding the optimum $\sigma_{max}$ in this case is a P problem, as shown by Algorithm \ref{Alg:solveLOP_n-2_1_1Zero_polynomial}. \\  

\begin{algorithm}
\hspace*{\algorithmicindent} \textbf{Input:} $A$, input matrix of the LOP of dimension $n$ such that $a_{ij}-a_{ji} + a_{jk}-a_{kj} = a_{ik}-a_{ki}\ \forall i,j,k = 1,2,\cdots n$. \\
\hspace*{\algorithmicindent} \textbf{Output:} $\sigma_{max}$.

\begin{algorithmic}[]
\caption{Polynomial algorithm that solves the LOP if $\hat{f}_{(n-2,1,1)} = 0$}\label{Alg:solveLOP_n-2_1_1Zero_polynomial}
\IF {$a_{12} - a_{21} > 0$}
	\STATE $\sigma_{max} = [1\ 2]$
\ELSE
	\STATE $\sigma_{max} = [2\ 1]$
\ENDIF
\FOR{$i = 3,4, \cdots, n$}
	\STATE $isFound = false$
	\FOR{$j = 1,2, \cdots, length(\sigma_{max})$}
		\IF {$a_{\sigma_{max}(j)i} - a_{i\sigma_{max}(j)} < 0$}
			\STATE insert element $i$ at position $j$ in $\sigma_{max}$
			\STATE $isFound = true$
			\STATE \textbf{break}
		\ENDIF
	\ENDFOR
	\IF {$isFound = false$}
		\STATE append element $i$ to $\sigma_{max}$
	\ENDIF
\ENDFOR
\end{algorithmic}
\end{algorithm}


\subsection{NP-hard component of the LOP}

The aim of this part is to show that an LOP with an objective function $f$ such that $\hat{f}_{(n-1,1)} = 0$ is equivalent to an LOP of dimension $n-1$.

If the function is a probability, that is, $f = p$, $\hat{p}_{(n-1,1)}$ encodes the information regarding order 1 marginals. If $\hat{p}_{(n-1,1)} = 0$, it means that this marginals have no relevant information, that is, that they are uniform distributions. If $f$ is a general function, not necessarily a probability, this idea can be translated in the following way:
\begin{equation}\label{eq::f_st_n-1_1_coeff0}
\sum\limits_{\sigma:\sigma(i) = j} f(\sigma) = \frac{1}{n}\sum\limits_{\sigma\in\Sigma_n} f(\sigma),
\end{equation}
for any $i,j=1,2,\cdots, n$.

Taking this information into account, Proposition \ref{prop::LOP_necessaryConditionA_n-1_1_zero} shows how having coefficient $(n-1,1)$ equal to 0 affects the shape of the elements of the input matrix.

\begin{proposition}\label{prop::LOP_necessaryConditionA_n-1_1_zero}

The objective function $f:\Sigma_n\longrightarrow\mathbb{R}$ of an LOP instance with an input matrix $A$ satisfies $\hat{f}_{(n-1,1)} = 0$ if and only if
$$\sum\limits_{j=1}^n \left(a_{ij}-a_{ji}\right) = 0,\ \forall i=1,\cdots n.$$
\end{proposition}

Lemma \ref{lem::LOP_equalities_cyclicTransformation} states a property that is satisfied if $\hat{f}_{(n-1,1)} = 0$.

\begin{lemma}\label{lem::LOP_equalities_cyclicTransformation}
If $f:\Sigma_n\longrightarrow\mathbb{R}$ is an LOP function such that $\hat{f}_{(n-1,1)} = 0$, then 
$$f([i_1\ i_2\ i_3\ \cdots i_n]) = f([i_n\ i_1\ i_2\ \cdots i_{n-1}]).$$
\end{lemma}

\begin{example}
This means that, under the conditions of Lemma \ref{lem::LOP_equalities_cyclicTransformation}, given a permutation, other $n-1$ permutations have the same objective function value. As an example, for $n=4$,
$$f([1\ 2\ 3\ 4]) = f([4\ 1\ 2\ 3]) = f([3\ 4\ 1\ 2]) = f([2\ 3\ 4\ 1]).$$
\end{example}

\begin{proposition}\label{prop::LOP_n-1_1_zero_optimum_reduction}
Let $f_A$ be an LOP function of dimension $n-1$ with an input matrix 
$$A = \begin{bmatrix}
a_{11} & a_{12} & \cdots & a_{1n-1}\\
a_{21} & a_{22} & \cdots & a_{2n-1}\\
\vdots & \vdots & \ddots & \vdots\\
a_{n-11} & a_{n-12} & \cdots & a_{n-1n-1}
\end{bmatrix}.$$ 
Let $f_{A'}$ be the LOP function with input matrix
$$A' = \begin{bmatrix}
a_{11} & a_{12} & \cdots & a_{1n-1} & a_{1n}\\
a_{21} & a_{22} & \cdots & a_{2n-1} & a_{2n}\\
\vdots & \vdots & \ddots & \vdots & \vdots\\
a_{n1} & a_{n2} & \cdots & a_{nn-1} & a_{nn}
\end{bmatrix},$$
where
$$a_{in} = -\sum_{j=1}^n (a_{ij} - a_{ji}),\ \forall i=1,2,\cdots,n-1,$$
and
$$a_{ni} = 0,\ \forall i = 1,2,\cdots,n.$$
Under these conditions, $\sigma' = [i_1\ i_2\cdots\ i_{n-1}\ n]$ is a global optimum of $f_{A'}$ if and only if $\sigma = [i_1\ i_2\ \cdots i_{n-1}]$ is a global optimum of $f_{A}$.
\end{proposition}

\begin{theorem}\label{theo::LOP_n-1_1_zero_NP_hard}
The problem composed by the LOP instances such that their objective function $f$ satisfies $\hat{f}_{(n-1,1)} = 0$ is NP-hard.
\end{theorem}

\begin{proof}
To prove this, we see that any LOP of dimension $n-1$ can be ``reduced'' to an LOP of dimension $n$ such that $\hat{f}_{(n-1,1)} = 0$. Let $A$ be the input matrix of an LOP of dimension $n-1$. If we define $A'$ as in the statement of Proposition \ref{prop::LOP_n-1_1_zero_optimum_reduction}, then solving the LOP of dimension $n-1$ without restrictions is equivalent to solving the LOP of dimension $n$ such that $\hat{f}_{(n-1,1)} = 0$.

\end{proof}


\subsection{Implications}

Propositions \ref{prop::LOP_necessaryConditionA_n-2_1_1_zero} and \ref{prop::LOP_necessaryConditionA_n-1_1_zero} give the exact conditions that an input matrix $A$ of the LOP must satisfy so that coefficient $(n-2,1,1)$ or $(n-1,1)$, respectively, is zero. Considering that the FT is linear and that there is a bijection between Fourier coefficients and permutation-based functions, it means that any LOP function $f$ can be written as
$$f = f_1 + f_2,$$
where $f_1$ has an input matrix as described in Proposition \ref{prop::LOP_necessaryConditionA_n-2_1_1_zero}, while $f_2$ has an input matrix as described in Proposition \ref{prop::LOP_necessaryConditionA_n-1_1_zero}. In addition, considering the subset of the LOP where the input matrix satisfies the conditions of Proposition \ref{prop::LOP_necessaryConditionA_n-2_1_1_zero}, we have seen that this is a P problem. On the other hand, the subset of the LOP for which the input matrix satisfies the conditions of Proposition \ref{prop::LOP_necessaryConditionA_n-1_1_zero} is NP-hard. 

As a final remark, since $f_1$ satisfies $\hat{f_1}_{(n-2,1,1)} = 0$, this function is only composed of the univariate information of $f$, while $f_2$ satisfies $\hat{f_2}_{(n-1,1)} = 0$, so it contains no univariate information of $f$ at all.

%
%

\section{Constructive algorithms}\label{sec::constructives}

In this section, the constructive algorithms used in the experimental part are presented. One of the algorithms is Becker's constructive \cite{becker1967helmstadtersche}, which is widely known in the literature, while the other three have been specifically designed to be used in this study.


\subsection{Constructive based on row/column subtractions}

The first constructive that we have designed has two versions. Both of them are rooted in the following theoretical results, which are related to the mean of the objective function values when the first elements of the permutations are fixed.

\begin{proposition}\label{prop::LOP_mean_partial_permu}
Let $S = \{\sigma: \sigma(1) = i_1, \sigma(2) = i_2, \cdots, \sigma(k) = i_{k}\ \ , \ k \le n\}$,  then
\begin{align*}
& \frac{1}{|S|}\sum_{\sigma\in S} f(\sigma) = \sum_{j\in \{1,2,\cdots,n\}-\{i_1\}} a_{i_1 j}  + \\& \sum_{j\in \{1,2,\cdots,n\}-\{i_1, i_2\}} a_{i_2 j} + \cdots + \\& \sum_{j\in \{1,2,\cdots,n\}-\{i_1, i_2, \cdots, i_{k}\}} a_{i_k j}  +\\& \frac{1}{2}\sum_{l,m\in \{1,2,\cdots,n\}-\{i_1,i_2, \cdots, i_k\}} a_{lm}.
\end{align*}
\end{proposition}

\begin{proposition}\label{prop::LOP_mean_difference_partial_permu}
Let $S_k = \{\sigma: \sigma(1) = i_1, \sigma(2) = i_2, \cdots, \sigma(k) = i_{k} \ , \ k \le n \}$ and $S_{k-1} = \{\sigma: \sigma(1) = i_1, \sigma(2) = i_2, \cdots, \sigma(k-1) = i_{k-1} \ , \ k \le n\}$ and
$$\mu_{k} = \frac{1}{|S_k|}\sum_{\sigma\in S_k} f(\sigma) \qquad\text{and}\qquad \mu_{k-1} = \frac{1}{|S_{k-1}|}\sum_{\sigma\in S_{k-1}} f(\sigma).$$
Then,
$$\mu_{k} - \mu_{k-1} = \frac{1}{2} \sum_{j\in \{1,2,\cdots,n\}-\{i_1, i_2, \cdots, i_{k}\}} (a_{i_k j} - a_{j i_k}). $$
\end{proposition}

The proof of both propositions is quite straightforward, so it has been omitted. Proposition \ref{prop::LOP_mean_partial_permu} can be proved by counting, for each pair of indices $i$ and $j$, how many permutations $\sigma\in S$ satisfy that $\sigma(i) < \sigma(j)$. On the other hand, Proposition \ref{prop::LOP_mean_difference_partial_permu} can be proved simply by applying the formula of Proposition \ref{prop::LOP_mean_partial_permu}.

By means of both propositions, one can define a procedure that, at each step $k$, constructs a solution $\sigma$ by fixing $\sigma(k) = i_k$, where $i_k$ is the index that maximizes the mean of the objective-function values assuming that $\sigma(1) = i_1$, $\sigma(2) = i_2$, $\cdots$, $\sigma(k-1) = i_{k-1}$ have been fixed. If these indices have been fixed, according to Proposition \ref{prop::LOP_mean_difference_partial_permu}, the difference between the mean obtained in the previous step and the new mean is fairly simple to compute. Instead of directly maximizing the value of the new mean, our constructive maximizes the difference between the previous and the new mean, which is equivalent. Algorithm \ref{alg::constructive_SS} describes this process.

\begin{algorithm}
\hspace*{\algorithmicindent} \textbf{Input:} $A$, input matrix of the LOP of dimension $n$. \\
\hspace*{\algorithmicindent} \textbf{Output:} $\sigma$.

\begin{algorithmic}[]
\caption{Simple constructive based on row/column subtractions}\label{alg::constructive_SS}
\STATE $I = \{1,2,\cdots,n\}$
\STATE $\sigma = ()$
\WHILE {$I$ is not empty}
	\STATE $q_i = \sum_{j\in I-\{i\}} (a_{ij} - a_{ji}),\ \forall i\in I$.
	\STATE $i_{max} = \argmax_{i\in I} \{q_i\}$
	\STATE append $i_{max}$ to $\sigma$
	\STATE delete $i_{max}$ from $I$
\ENDWHILE
\end{algorithmic}
\end{algorithm}

A more sophisticated version of the procedure can be designed by observing that Proposition \ref{prop::LOP_mean_partial_permu} can be extended to indices that are fixed at the end of the permutation, as follows:

\begin{proposition}\label{prop::LOP_mean_partial_permu_sophisticated}
Let $S_1 = \{\sigma: \sigma(1) = i_1, \sigma(2) = i_2, \cdots, \sigma(k) = i_{k} \ , \ k \le n\}$ and $S_2 = \{\sigma: \sigma(n) = j_1, \sigma(n-1) = j_2, \cdots, \sigma(n-l+1) = j_{l} \ , \ l \ge 1\}$,  then
\begin{align*}
& \frac{1}{|S_1\cap S_2|}\sum_{\sigma\in S_1\cap S_2} f(\sigma) = \sum_{j\in \{1,2,\cdots,n\}-\{i_1\}} a_{i_1 j}  + \\& \sum_{j\in \{1,2,\cdots,n\}-\{i_1, i_2\}} a_{i_2 j} + \cdots + \\& \sum_{j\in \{1,2,\cdots,n\}-\{i_1, i_2, \cdots, i_{k}\}} + \sum_{i\in \{1,2,\cdots,n\}-\{i_1,i_2,\cdots,i_k\}} a_{i j_1} + \\& \cdots + \sum_{i\in \{1,2,\cdots,n\}-\{i_1,i_2,\cdots,i_k,j_1,\cdots, j_{l-1}\}} a_{i j_l} +
\\& \frac{1}{2}\sum_{i,j\in \{1,2,\cdots,n\}-\{i_1,i_2, \cdots, i_k, j_1, j_2,\cdots, j_l\}} a_{ij}.
\end{align*}
\end{proposition}

Proceeding analogously to the simpler version, Algorithm \ref{alg::constructive_S} can be proposed.

\begin{algorithm}
\hspace*{\algorithmicindent} \textbf{Input:} $A$, input matrix of the LOP of dimension $n$. \\
\hspace*{\algorithmicindent} \textbf{Output:} $\sigma$.

\begin{algorithmic}[]
\caption{Constructive based on row/column subtractions}\label{alg::constructive_S}
\STATE $I = \{1,2,\cdots,n\}$
\STATE $\sigma = (-1)$
\WHILE {$I$ is not empty}
	\STATE $q_i = \sum_{j\in I-\{i\}} (a_{ij} - a_{ji}),\ \forall i\in I$.
	\STATE $i_1 = \argmax \{q_i: i\in I\}$
	\STATE $i_2 = \argmax \{-q_i: i\in I\}$
	\IF {$q_{i_1} > -q_{i_2}$}
		\STATE insert $i_1$ in $\sigma$ immediately before the $-1$
		\STATE delete $i_1$ from $I$
	\ELSE 
		\STATE insert $i_2$ in $\sigma$ immediately after the $-1$
		\STATE delete $i_2$ from $I$
	\ENDIF
\ENDWHILE
\STATE delete the -1 from $\sigma$
\end{algorithmic}
\end{algorithm}


\subsection{Constructive based on the matrix of univariate information}

Instead of just considering the mean values when the index is fixed at the beginning of the permutation or at the end, as is done in the previous constructive, the idea of this constructive consists in obtaining a solution taking into account all of the means of the objective function values when $\sigma(i) = j$, for any $i$ and $j$. The first step in this direction can be taken through Propositions \ref{prop::LOP_meanValues_i_j} and \ref{prop::LOP_meanValues_i_j_difference} below:

\begin{proposition}\label{prop::LOP_meanValues_i_j}
Given two indices $i,j\in\{1,2,\cdots, n\}$, the mean objective-function value of the permutations $\sigma$ such that $\sigma(i) = j$ is the following:
\begin{align*}
& \frac{1}{(n-1)!} \sum_{\sigma:\sigma(i) = j} f(\sigma) = \frac{i-1}{n-1}\sum_{k\ne j}a_{kj} + \frac{n-i}{n-1}\sum_{k\ne j}a_{jk} + \\& \frac{1}{2}\sum_{\substack{k\ne l\\k,l\ne j}} a_{kl}.
\end{align*}
\end{proposition}

\begin{proposition}\label{prop::LOP_meanValues_i_j_difference}
Given
$$\mu_{ij}=\frac{1}{(n-1)!} \sum_{\sigma:\sigma(i) = j} f(\sigma)$$
and
$$\mu_{i+1j}=\frac{1}{(n-1)!} \sum_{\sigma:\sigma(i+1) = j} f(\sigma),$$
then,
$$\mu_{i+1 j} - \mu_{i j} = \frac{1}{n-1}\sum_{k\ne j} (a_{kj} - a_{jk}).$$
\end{proposition}

As happens with the propositions of the previous constructive, this propositions are quite straightforward to prove, so the prove is again omitted. Observe that the difference between $\mu_{ij}$ and $\mu_{i+1 j}$ is independent of $i$. This fact has two main consequences. Firstly, for a given $j$, the highest mean is always reached when $\sigma(1) = j$ or when $\sigma(n) = j$. Secondly, since the difference is independent of $i$, the set $\{\mu_{ij}\ |\ i,j = 1,2,\cdots,n\}$ of mean values is very easy to be computed. The constructive algorithm that we propose is based on the matrix of means $[\mu_{ij}]$, and solves the following problem:
\begin{equation}\label{eq::LAP_matrix_of_means}
\sigma = \argmax_{\omega\in\Sigma_n} \sum_{i=1,2,\cdots,n} \mu_{i \omega(i)}, 
\end{equation}
  which is a Linear Assignment Problem \cite{commander2005survey} and can be solved efficiently in polynomial time. In our procedure, instead of solving the problem with matrix $[\mu_{ij}]$, we solve it with a matrix $M$ defined as follows:
\begin{align*}
m_{1j} = & \sum_{j=1}^n (a_{1j} - a_{j1}),\ \forall j = 1,2,\cdots,n. \\
m_{nj} = & \sum_{j=1}^n (a_{nj} - a_{jn}),\ \forall j = 1,2,\cdots,n. \\
m_{ij} = & m_{1j} + (i-1)\frac{m_{nj} - m_{1j}}{n-1},\ \forall i,j\ne 1.
\end{align*}
Following a similar reasoning as in the previous section, one can check that solving the Linear Assignment Problem with matrix $M$ is equivalent to solving equation (\ref{eq::LAP_matrix_of_means}).

\begin{algorithm}
\hspace*{\algorithmicindent} \textbf{Input:} $A$, input matrix of the LOP of dimension $n$. \\
\hspace*{\algorithmicindent} \textbf{Output:} a matrix $M$ of dimension $n$.

\begin{algorithmic}[]
\caption{Constructive based on the matrix of univariate information}\label{alg::matrix}
\STATE $m_{1j} = \sum_{j}(a_{1j} - a_{j1}),\ \forall j = 1,2,\cdots,n$.
\STATE $m_{n1} = -m_{1j},\ \forall j=1,2,\cdots,n$.
\STATE $m_{ij} = m_{1j} + (i-1)\frac{m_{nj}-m_{1j}}{n-1},   \forall i,j\ne 1$.
\STATE compute $\sigma = \argmax_{\omega\in\Sigma_n} \sum m_{i\omega(i)}$.
\end{algorithmic}
\end{algorithm}

%
%

\section{Experimental study}\label{sec::experimental_study}

In Section \ref{sec::decomposition_LOP}, we showed how an LOP can be expressed as the sum of two subproblems, one of which is P, while the other is NP-hard. The P subproblem is composed of all of the univariate information, while the NP-hard subproblem contains the rest of the information when the univariate information is subtracted. The aim of the experimental part of this study is to observe the consequences of this decomposition at a practical level. Particularly, we see how the constructives previously described behave depending on the contribution of the each of the components. For this sake, we generate pairs of matrices, $A_{(n-1,1)}$ and $A_{(n-2,1,1)}$, such that the first is associated with the P component, while the second is associated with the NP-hard component. Considering reduced values of the dimension $n$, we study how the different constructives perform when we try to solve the LOP whose input matrix is
$$ A_{(n-1,1)} + \epsilon A_{(n-2,1,1)}, $$
with different values of $\epsilon$, where $\epsilon$ represents the weight that is given to the NP-hard component.

\subsection{Experimental design}

The steps of the experiment are the following:

\begin{enumerate}
	\item Generate a pair of ``random'' input matrices $A_{(n-1,1)}$ and $A_{(n-2,1,1)}$ associated with the P and NP-hard components, respectively.
	\item Compute and store the exact solution of the LOP whose input matrix is $$ A_{(n-1,1)} + \epsilon A_{(n-2,1,1)}, $$ with different values of $\epsilon$.
	\item Apply the four constructive algorithms (Becker's, the two versions of the constructive based on row/column subtractions and the constructive based on the matrix of univariate information) to $$ A_{(n-1,1)} + \epsilon A_{(n-2,1,1)}, $$ for each $\epsilon$ fixed in Step 2.
	\item Compute the error of the solutions given by the constructives with respect to the exact solution according to the following metric:
	$$\frac{|f(\sigma)-f(\sigma_{max})|}{|f(\sigma_{max}) - f(\sigma_{min})|},$$
	where $\sigma_{max}$ is the global optimum and $\sigma_{min}$ is the global minimum.
\end{enumerate}

We run 20 repetitions of this experiment for each of the dimensions $n = 10,11$. The values of $\epsilon$ have been chosen according to a logarithmic scale and are the following:
$$\epsilon = 0, 10^{-2}, 10^{-1.75}, 10^{-1.5}, 10^{-1.25}, \cdots, 10^2, 10^{2.25}, 10^{2.5}.$$
Regarding Step 3, Becker's constructive can not be applied when the sum of all the elements of a row/column is negative. So, before applying it to any matrix $B = A_{(n-1,1)} + \epsilon A_{(n-2,1,1)}$, we subtract the minimum value of $C$ to all of the elements of $C$ in order to avoid negative values.

\subsubsection{How to generate $A_{(n-1,1)}$}

$A_{(n-1,1)} = [a_{ij}]$ is generated so that the condition of the statement of Proposition \ref{prop::LOP_necessaryConditionA_n-2_1_1_zero} is satisfied. For this sake, we introduce the matrix of differences $D$ for which $d_{ij} = a_{ij}-a_{ji}$. The steps to generate $A_{(n-1,1)}$ are the following:

\begin{enumerate}
	\item Generate a random permutation $\sigma\in\Sigma_n$.
	\item For $i = 1,2,\cdots, n-1$, do randomly one of these steps: 
		\begin{enumerate}
			\item Assign $a_{\sigma(i)\sigma(i+1)}$ and $a_{\sigma(i+1)\sigma(i)}$ a random value respectively. Then, set $d_{\sigma(i)\sigma(i+1)} = a_{\sigma(i)\sigma(i+1)} - a_{\sigma(i+1)\sigma(i)}$ and $d_{\sigma(i+1)\sigma(i)} = a_{\sigma(i+1)\sigma(i)} - a_{\sigma(i)\sigma(i+1)}$.
			\item  Assign $a_{\sigma(i)\sigma(i+1)}$ and $d_{\sigma(i)\sigma(i+1)}$ a random value respectively. Then, set $d_{\sigma(i+1)\sigma(i)} = -d_{\sigma(i)\sigma(i+1)}$ and $a_{\sigma(i+1)\sigma(i)} = d_{\sigma(i+1)\sigma(i)} + a_{\sigma(i)\sigma(i+1)}$.
			\item Assign $d_{\sigma(i)\sigma(i+1)}$ and $a_{\sigma(i+1)\sigma(i)}$ a random value. Then, set $d_{\sigma(i+1)\sigma(i)} = -d_{\sigma(i)\sigma(i+1)} $ and $a_{\sigma(i)\sigma(i+1)} = d_{\sigma(i)\sigma(i+1)} + a_{\sigma(i+1)\sigma(i)}$.
		\end{enumerate}
	\item Set the rest of the values of $D$, so that $d_{ij} + d_{jk} = d_{ik}$ and $d_{ij} = -d_{ji}$.
	\item Until $A_{(n-1,1)}$ is completely defined, choose a random couple of indices $i$ and $j$ for which $a_{ij}$ is not defined. Set $a_{ij}$ following a uniform distribution $U(-1,1)$ and $a_{ji} = d_{ji} + a_{ij}$. 
\end{enumerate}

All of the random numbers are generated following a uniform distribution $U(-1,1)$. The sense of the range of options of Step 2 is to avoid creating unnecessary structures that could bias the experiment. For example, if the first substep of Step 2 was the only choice, then, it would always happen that $a_{\sigma(i)\sigma(i+1)},a_{\sigma(i+1)\sigma(i)}\in(-1,1)$, while $d_{\sigma(i)\sigma(i+1)}, d_{\sigma(i+1)\sigma(i)}\in (-2,2)$.

\subsubsection{How to generate $A_{(n-2,1,1)}$}

$A_{(n-2,1,1)} = [a_{ij}]$ is generated so that the condition of the statement of Proposition \ref{prop::LOP_necessaryConditionA_n-1_1_zero} is satisfied. For this sake, we introduce the matrix of differences $D$ for which $d_{ij} = a_{ij}-a_{ji}$. The steps to generate $A_{(n-1,1)}$ are the following:

\begin{enumerate}
	\item Choose two random indices $i$ and $j$ such that $a_{ij}$ has not been assigned.
	\item Do randomly one of these steps:
		\begin{enumerate}
			\item Assign $a_{ij}$ and $a_{ji}$ random values. Then, set $d_{ij} = a_{ij}- a_{ji}$ and $d_{ji} = a_{ji} - a_{ij}$.
			\item Assign randomly $d_{ij}$ and $a_{ij}$. Then, set $d_{ji} = -d_{ij}$ and $a_{ji} = -a_{ij}$.
		\end{enumerate}
	\item If there is a row $i$ in $D$ such that all of the elements except one, $d_{ik}$ have been assigned, set $d_{ik} = -\sum_{j\ne k} d_{ij}$. Then, set $d_{ki} = -d_{ik}$. Then, do randomly one of these steps:
		\begin{enumerate}
			\item Assign $a_{ik}$ a random value. Then, set $a_{ki} = d_{ki} + a_{ik}$.
			\item Assign $a_{ki}$ a random value. Then, set $a_{ik} = d_{ik} + a_{ki}$.
		\end{enumerate}
	\item Repeat Step 3 until there is no row of $D$ with a single element that has not been assigned any value.
	\item Repeat all of the steps until $A_{(n-2,1,1)}$ has been assigned all of the values.
\end{enumerate}

All of the random numbers are generated following a uniform distribution $U(-1,1)$.

\subsection{Experimental results}

Table \ref{tab::results_summary} shows the mean value of the errors of the solutions found by the constructives (with respect to the optimum) after 20 repetitions. SS and S refer to the simpler and more sophisticated version of the constructive based on row/column subtractions, while CM refers to the constructive based on the matrix of univariate information.

\begin{table*}
\small
\centering
\caption{This table shows the errors of the solutions found by the constructives with respect to the optimum for dimensions $n=10,11$ and with a number of values of $\epsilon$.}\label{tab::results_summary}
 \begin{tabular}{r| r r r r | r r r r }
& \multicolumn{4}{|c |}{Errors when $n = 10$} & \multicolumn{4}{c}{Errors when $n = 11$} \\ 
\hline
\multirow{2}{*}{$\epsilon\quad$} & \multirow{2}{*}{Becker} & \multirow{2}{*}{SS} & \multirow{2}{*}{S} & \multirow{2}{*}{CM} & \multirow{2}{*}{Becker} & \multirow{2}{*}{SS} & \multirow{2}{*}{S} & \multirow{2}{*}{CM} \\
& & & & & & & &\\
\hline
& & & & & \\
0.000 & 0.002 & 0.000 & 0.000 & 0.000 & 0.001 & 0.000 & 0.000 & 0.000\\
0.010 & 0.002 & 0.000 & 0.000 & 0.000 & 0.001 & 0.000 & 0.000 & 0.000\\
0.018 & 0.002 & 0.000 & 0.000 & 0.000 & 0.001 & 0.000 & 0.000 & 0.000\\
0.032 & 0.002 & 0.000 & 0.000 & 0.000 & 0.001 & 0.000 & 0.000 & 0.000\\
0.056 & 0.003 & 0.000 & 0.000 & 0.001 & 0.002 & 0.000 & 0.000 & 0.001\\
0.100 & 0.004 & 0.002 & 0.001 & 0.003 & 0.003 & 0.002 & 0.001 & 0.003\\
0.178 & 0.007 & 0.005 & 0.002 & 0.009 & 0.007 & 0.006 & 0.005 & 0.009\\
0.316 & 0.014 & 0.014 & 0.008 & 0.022 & 0.018 & 0.018 & 0.015 & 0.024\\
0.562 & 0.026 & 0.026 & 0.013 & 0.051 & 0.042 & 0.041 & 0.027 & 0.055\\
1.000 & 0.047 & 0.040 & 0.027 & 0.104 & 0.062 & 0.062 & 0.039 & 0.109\\
1.778 & 0.080 & 0.074 & 0.051 & 0.179 & 0.079 & 0.082 & 0.059 & 0.183\\
3.162 & 0.098 & 0.091 & 0.049 & 0.262 & 0.094 & 0.092 & 0.076 & 0.270\\
5.623 & 0.107 & 0.095 & 0.066 & 0.338 & 0.108 & 0.110 & 0.083 & 0.348\\
10.000 & 0.111 & 0.102 & 0.064 & 0.397 & 0.108 & 0.114 & 0.071 & 0.408\\
17.783 & 0.107 & 0.094 & 0.071 & 0.439 & 0.091 & 0.099 & 0.062 & 0.449\\
31.623 & 0.102 & 0.097 & 0.069 & 0.467 & 0.086 & 0.090 & 0.056 & 0.475\\
56.234 & 0.092 & 0.091 & 0.067 & 0.484 & 0.083 & 0.090 & 0.051 & 0.491\\
100.000 & 0.091 & 0.089 & 0.064 & 0.494 & 0.085 & 0.093 & 0.049 & 0.500\\
177.828 & 0.092 & 0.087 & 0.071 & 0.499 & 0.081 & 0.091 & 0.047 & 0.505\\
316.228 & 0.091 & 0.086 & 0.067 & 0.503 & 0.080 & 0.084 & 0.046 & 0.508\\
\end{tabular}
\end{table*}

It can be clearly seen how, as the value of $\epsilon$ increases, the behavior of the constructives degrades, having small errors for values of $\epsilon$ close to 0, rapidly increasing when this value gets further from zero and tending to a stability for greater values of $\epsilon$. This behavior is better reflected by Figure \ref{fig::evolution_errors_n11}, where the evolution of the mean errors is represented as $\epsilon$ varies for a dimension of $n=11$. 

\begin{figure*}[htbp]
\centerline{\includegraphics[scale=0.5]{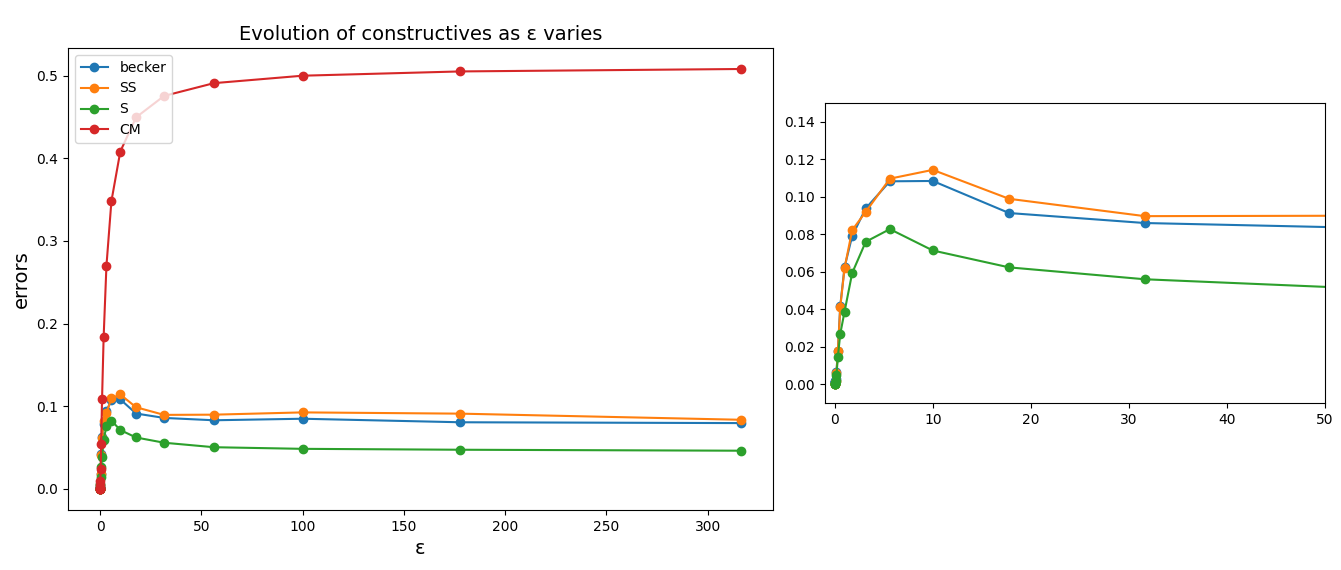}}
\caption{Evolution of the mean value of the errors when the different constructive algorithms are applied, as $\epsilon$ varies and when the dimension is $n=11$. The diagram at the right is a zoomed image of Becker's constructive and the two constructives based on row/column subtractions.}
\label{fig::evolution_errors_n11}
\end{figure*}

For the smallest values of $\epsilon$, Becker presents the worst performance, which could be explained by the following reasoning: even though Becker uses information about the ratio between row/column sums, and this is related to univariate information, this does not explicitly represent the mean values that are used in the rest of constructives. On the other hand, for bigger values of $\epsilon$, the constructive based on the matrix of mean values turns out to be the worst. This is the only constructive that fixes all of the elements $\sigma(1)$, $\sigma(2)$, $\cdots$, $\sigma(n)$ at the same time, instead of fixing them step by step. Probably the lack of flexibility, which is caused by the fact that the mean values are not recalculated after fixing each element, is the reason for this behavior. That is, the influence of the information gathered from the marginals of order 1 (univariate information) is too strong, and its performance gets worse as the weight of the other marginals increases.

%
%

\section{Conclusions and future work}\label{sec::conclusions}

In this paper, we prove, with the aid of the FT, how the LOP can be expressed as a sum of two subproblems, one of which is P, and contains all of the univariate information, while the other is NP-hard and contains no relevant univariate information at all. We experimentally study how giving different weights to the NP-hard component affects the behavior of a number of constructive algorithms that mainly take into account univariate information. In these experiments, we observe that the performance of these algorithms drastically degrades as the problem transits from P to NP-hardness, as in a phase transition.

When the NP-hard component has no effect (it is given a weight of 0), the worst performance is presented by the Becker's constructive. On the other hand, when the NP-hard component gains more influence, the last constructive that we propose is the one that works worst. This last algorithm, contrarily to the rest, which construct the solution step by step, constructs the solution having a single glance at the global information. We think that it could be interesting to combine the best characteristics of both types of heuristics, that is, taking into account the global information in order to construct the solution step by step while we keep recomputing the information at each step.   

Another interesting open question naturally arises from this study. We have seen how to define input matrices for the LOP in such a way that they are associated with the P or NP-hard components. However, it would be interesting to take the opposite path. That is, given an LOP instance defined by an input matrix $A$, can we find two matrices $B$ and $C$ such that $B$ is associated with the P component, $C$ is associate with the NP-hard component and their sum is $A$?

%
%

\section*{Acknowledgment}


This work is supported by the Basque Government (BERC 2022-2025 and IT1504-22) and by the Spanish Ministry of Economy and Competitiveness MINECO (projects PID2019-104966GB-I00 and PID2019-106453GA-I00/AEI/10.13039/501100011033). Jose A. Lozano acknowledges support by the Spanish Ministry of Science, Innovation and Universities through BCAM Severo Ochoa accreditation (SEV-2017-0718). Anne Elorza holds a predoctoral grant (ref. PIF17/293) from the University of the Basque Country.

%
%

\appendix 
\subsection{Proofs of Section \ref{sec::decomposition_LOP}}\label{app:proofs_decomposition_LOP}


\begin{proof}[Proof of Proposition \ref{prop::LOP_necessaryConditionA_n-2_1_1_zero}]
Let us consider two indices $i$ and $j$. We assume, without loss of generality that $i,j\ne 1,2, n$ and $i < j$. Then, 
\begin{align*}
a_{ij} - a_{ji} = & f([i\ j\ 1\ 2\ \cdots\ i-1\ i+1\ \cdots\ j-1\ j+1\ \cdots \ n]) \\
-&f([j\ i\ 1\ 2\ \cdots\ i-1\ i+1\ \cdots\ j-1\ j+1\ \cdots \ n]).
\end{align*}

By taking into account equation (\ref{eq::function_only_2lowestCoefficients}), this expression becomes
\begin{align*}
& a_{ij} - a_{ji} = \frac{1}{n!} \Bigg( \sum\limits_{\sigma:\sigma(1) = i} f(\sigma) + \sum\limits_{\sigma:\sigma(2) = j} f(\sigma) + \sum\limits_{\sigma:\sigma(3) = 1} f(\sigma) + \\
& \cdots + \sum\limits_{\sigma:\sigma(n) = n} f(\sigma) \Bigg) - 
\frac{1}{n!} \Bigg( \sum\limits_{\sigma:\sigma(1) = j} f(\sigma) + \sum\limits_{\sigma:\sigma(2) = i} f(\sigma) + \\ & \sum\limits_{\sigma:\sigma(3) = 1} f(\sigma) + \cdots + \sum\limits_{\sigma:\sigma(n) = n} f(\sigma) \Bigg) = 
 \frac{1}{n!} \Bigg( \sum\limits_{\sigma:\sigma(1) = i} f(\sigma) +\\& \sum\limits_{\sigma:\sigma(2) = j} f(\sigma) - \sum\limits_{\sigma:\sigma(1) = j} f(\sigma) - \sum\limits_{\sigma:\sigma(2) = i} f(\sigma) \Bigg).\\
\end{align*}

Then,
\begin{align*}
& a_{ij} - a_{ji} + a_{jk} - a_{kj} = \frac{1}{n!} \Bigg( \sum\limits_{\sigma:\sigma(1) = i} f(\sigma) + \sum\limits_{\sigma:\sigma(2) = j} f(\sigma) - \\& \sum\limits_{\sigma:\sigma(1) = j} f(\sigma) - \sum\limits_{\sigma:\sigma(2) = i} f(\sigma) \Bigg) + 
 \frac{1}{n!} \Bigg( \sum\limits_{\sigma:\sigma(1) = j} f(\sigma) + \\& \sum\limits_{\sigma:\sigma(2) = k} f(\sigma) - \sum\limits_{\sigma:\sigma(1) = k} f(\sigma) - \sum\limits_{\sigma:\sigma(2) = j} f(\sigma) \Bigg) = \\
& \frac{1}{n!} \Bigg( \sum\limits_{\sigma:\sigma(1) = i} f(\sigma) + \sum\limits_{\sigma:\sigma(2) = k} f(\sigma) - \sum\limits_{\sigma:\sigma(1) = k} f(\sigma) - \\& \sum\limits_{\sigma:\sigma(2) = i} f(\sigma) \Bigg) = 
 a_{ik} - a_{ki}.
\end{align*}

\end{proof}

\begin{proof}[Proof of Proposition \ref{prop:LOP_n-2_1_1_zero_swap}]

Taking into account that 
\small
$$ \sigma = [\sigma(1)\ \cdots \sigma(i-1)\ \sigma(i)\ \sigma(i+1)\ \cdots \sigma(j-1)\ \sigma(j)\ \sigma(j+1)\ \cdots \sigma(n)]$$
\normalsize and \small
\begin{align*}
&\sigma_{ij}\circ\sigma = \\&[\sigma(1)\ \cdots \sigma(i-1)\ \sigma(j)\ \sigma(i+1)\ \cdots \sigma(j-1)\ \sigma(i)\ \sigma(j+1)\ \cdots \sigma(n)],
\end{align*}
\normalsize evaluating $f(\sigma)$ and $f(\sigma_{ij}\circ\sigma)$, it can be seen that

\begin{align*}
	& f(\sigma) - f(\sigma_{ij}\circ\sigma) = 
	a_{\sigma(i)\sigma(i+1)} + a_{\sigma(i)\sigma(i+2)} + \cdots + \\& a_{\sigma(i)\sigma(j-1)} + a_{\sigma(i)\sigma(j)} + a_{\sigma(i+1)\sigma(j)} + a_{\sigma(i+2)\sigma(j)} + \cdots + \\& 
	a_{\sigma(j-1)\sigma(j)} - ( a_{\sigma(j)\sigma(i+1)} + a_{\sigma(j)\sigma(i+2)} + \cdots + a_{\sigma(j)\sigma(j-1)} + \\& a_{\sigma(j)\sigma(i)} + a_{\sigma(i+1)\sigma(i)} + 
	a_{\sigma(i+2)\sigma(i)} + \cdots + a_{\sigma(j-1)\sigma(i)}  ). 
\end{align*}

In this expression, $a_{\sigma(i)\sigma(k)}$ and $a_{\sigma(k)\sigma(j)}$ appear exactly once, for any $k=i+1, i+2, \cdots, j-1$. Taking into account Proposition \ref{prop::LOP_necessaryConditionA_n-2_1_1_zero}, $a_{\sigma(i)\sigma(k)} + a_{\sigma(k)\sigma(j)} = a_{\sigma(i)\sigma(j)}$. Therefore, the expression above can be expressed as follows:
\begin{align*}
	& a_{\sigma(i)\sigma(j)} + a_{\sigma(i)\sigma(j)} + \cdots + a_{\sigma(i)\sigma(j)} - ( a_{\sigma(j)\sigma(i+1)} + \\& a_{\sigma(j)\sigma(i+2)} + \cdots + a_{\sigma(j)\sigma(j-1)} +  
	 a_{\sigma(j)\sigma(i)} + a_{\sigma(i+1)\sigma(i)} + \\& a_{\sigma(i+2)\sigma(i)} + \cdots + a_{\sigma(j-1)\sigma(i)}  ) = (j-i)a_{\sigma(i)\sigma(j)} - \\& ( a_{\sigma(j)\sigma(i+1)} + 
	 a_{\sigma(j)\sigma(i+2)} + \cdots + a_{\sigma(j)\sigma(j-1)} + a_{\sigma(j)\sigma(i)} + \\& a_{\sigma(i+1)\sigma(i)} + a_{\sigma(i+2)\sigma(i)} + \cdots + a_{\sigma(j-1)\sigma(i)}  ). 
\end{align*}
On the other hand, in the part of this expression that is between the parenthesis, $a_{\sigma(j)\sigma(k)}$ and $a_{\sigma(k)\sigma(i)}$ appear exactly once for $k=i+1, i+2, \cdots, j-1$. Therefore, the expression can be simplified as follows:
\begin{align*}
	& f(\sigma) - f(\sigma_{ij}\circ\sigma) = 
	 (j-i)a_{\sigma(i)\sigma(j)} - ( a_{\sigma(j)\sigma(i)} + \\& a_{\sigma(j)\sigma(i)} + \cdots + a_{\sigma(j)\sigma(i)} ) = 
	 (j-i)a_{\sigma(i)\sigma(j)} - \\& (j-i)a_{\sigma(i)\sigma(j)} = (j-i)(a_{\sigma(i)\sigma(j)} - a_{\sigma(j)\sigma(i)}). 
\end{align*}

\end{proof}


\begin{proof}[Proof of Proposition \ref{prop::LOP_necessaryConditionA_n-1_1_zero}]

Assume that $f$ is an LOP function such that $\hat{f}_{(n-1,1)} = 0$. Consider the following ``marginals'':
$$ \sum\limits_{\sigma:\sigma(1) = i} f(\sigma) \qquad\text{and}\qquad \sum\limits_{\sigma:\sigma(n) = i} f(\sigma).$$

According to equation (\ref{eq::f_st_n-1_1_coeff0}),
$$ \sum\limits_{\sigma:\sigma(1) = i} f(\sigma) = \sum\limits_{\sigma:\sigma(n) = i} f(\sigma),$$
then,
\begin{align*}
 & (n-1)!(a_{i2} + a_{i3} + \cdots a_{in}) + \frac{(n-1)!}{2}\left(\sum\limits_{k,l\ne 1} a_{kl} \right) = \\& (n-1)!(a_{2i} + a_{3i} + \cdots a_{ni}) + \frac{(n-1)!}{2}\left(\sum\limits_{k,l\ne 1} a_{kl} \right).
\end{align*}
Therefore,
\begin{align*}
& (n-1)!(a_{i2} + a_{i3} + \cdots a_{in}) = \\& (n-1)!(a_{2i} + a_{3i} + \cdots a_{ni}).
\end{align*}
And, finally,
$$ a_{i2} + a_{i3} + \cdots + a_{in} - (a_{2i} + a_{3i} + \cdots + a_{ni}) = 0.$$
\end{proof}

\begin{proof}[Proof of Lemma \ref{lem::LOP_equalities_cyclicTransformation}]
Let $A$ be an input matrix that generates $f$ and that satisfies the property of Proposition \ref{prop::LOP_necessaryConditionA_n-1_1_zero}.
\begin{align*}
f([i_1\ i_2\ i_3\ \cdots i_n]) = & a_{i_1i_2} + a_{i_1i_3} + \cdots + a_{i_1i_n} + \\
	                         & a_{i_2i_3} + a_{i_2i_4} + \cdots + a_{i_2i_n} + \\
	                         & \cdots +\\
	                         & a_{i_{n-1}i_n} = \\
	                         & a_{i_1i_2} + a_{i_1i_3} + \cdots + a_{i_1i_{n-1}} + \\
	                         & a_{i_2i_3} + a_{i_2i_4} + \cdots + a_{i_2i_{n-1}} + \\
	                         & \cdots +\\
	                         & a_{i_{n-2}i_{n-1}} +\\
	                         & a_{i_1 i_n} + a_{i_2 i_n} + \cdots + a_{i_{n-1}i_n}.
\end{align*}

\begin{align*}
f([i_n\ i_1\ i_2\ \cdots i_{n-1}]) = & a_{i_ni_1} + a_{i_ni_2} + \cdots + a_{i_ni_{n-1}} + \\
				  & a_{i_1i_2} + a_{i_1i_3} + \cdots + a_{i_1i_{n-1}} + \\
	                         & \cdots +\\
	                         & a_{i_{n-2}i_{n-1}}.
\end{align*}

Then,
\begin{align*}
& f([i_1\ i_2\ i_3\ \cdots i_n]) - f([i_n\ i_1\ i_2\ \cdots i_{n-1}]) = a_{i_1 i_n} + a_{i_2 i_n} + \\&  \cdots +a_{i_{n-1}i_n} - (a_{i_ni_1} + a_{i_ni_2} + \cdots + a_{i_ni_{n-1}}),
\end{align*}
which, according to Proposition \ref{prop::LOP_necessaryConditionA_n-1_1_zero}, is 0.
\end{proof}

\begin{proof}[Proof of Proposition \ref{prop::LOP_n-1_1_zero_optimum_reduction}]
Let $\sigma= [i_1\ i_2\ \cdots i_{n-1}\ n]$ be a global optimum of $f_{A'}$ (taking into account Proposition \ref{lem::LOP_equalities_cyclicTransformation}, there always exists a global optimum $\sigma$ such that $\sigma(n) = n$). Then,
$$f_{A'}([i_1\ i_2\ \cdots\ i_{n-1}\ n]) \ge f_{A'}([j_1\ j_2\ \cdots\ j_{n-1}\ n]).$$

By developing both expressions, as in the proof of Lemma \ref{lem::LOP_equalities_cyclicTransformation}, the following inequality is obtained:
\begin{align*}
 & a_{i_1i_2} + a_{i_1i_2} + \cdots + a_{i_1i_{n-1}} + a_{i_2i_3} + \cdots + a_{i_2i_{n-1}} + \cdots + \\ & a_{i_{n-2}i_{n-1}} \ge 
  a_{j_1j_2} + a_{j_1j_2} + \cdots + a_{j_1j_{n-1}} + a_{j_2j_3} + \cdots + \\& a_{j_2j_{n-1}} + \cdots + a_{j_{n-2}j_{n-1}}. 
\end{align*}

This is equivalent to 
$$f_A([i_1\ i_2\ \cdots\ i_{n-1}]) \ge f_A([j_1\ j_2\ \cdots\ j_{n-1}]).$$

Since this happens for any permutation $[j_1\ j_2\ \cdots j_{n-1}]$, then $[i_1\ i_2\ \cdots\ i_{n-1}]$ is a global optimum.\\

To prove the reciprocal implication, assume that $[i_1\ i_2\ \cdots\ i_{n-1}]$ is a global optimum of $f_A$. Then, for any permutation $[j_1\ j_2\ \cdots\ j_{n-1}]\in\Sigma_{n-1}$,
$$ f_A([i_1\ i_2\ \cdots\ i_{n-1}]) \ge f_A([j_1\ j_2\ \cdots\ j_{n-1}]).$$

Then, if both objective function values are developed it is immediate that
$$f_{A'}([i_1\ i_2\ \cdots\ i_{n-1}\ n]) \ge f_{A'}([j_1\ j_2\ \cdots\ j_{n-1}\ n]).$$

Taking into account Lemma \ref{lem::LOP_equalities_cyclicTransformation},
\begin{align*}
 & f_{A'}([i_1\ i_2\ \cdots\ i_{n-1}\ n]) \ge f_{A'}([j_1\ j_2\ \cdots\ j_{n-1}\ n]) = \\& f_{A'}([j_2\ j_3\ \cdots\ j_{n-1}\ n\ j_1]) = f_{A'}([j_3\ j_4\ \cdots\ j_{n-1}\ n\ j_1\ j_2]) = \\& \cdots = f_{A'}([n\ j_1\ j_2\ j_3\ \cdots\ j_{n-1}]).
\end{align*}

Since the right side of the inequality contains any possible permutation,
$$f_{A'}([i_1\ i_2\ \cdots i_{n-1}\ n])\ge f_{A'}(\omega),\quad \forall\omega\in\Sigma_n.$$
Then, $[i_1\ i_2\ \cdots i_{n-1}\ n]$ is a global optimum of $f_{A'}$

\end{proof}

%
%

\bibliographystyle{IEEEtran}
\bibliography{IEEEabrv, bib_CEC2022}


\end{document}